\documentclass[a4paper]{article}

\usepackage[colorlinks=true, linkcolor={blue}, citecolor={red}, urlcolor={violet}]{hyperref}
  
\usepackage[english]{babel}
\usepackage[utf8]{inputenc}
\usepackage{amsmath, amssymb, amsthm}
\usepackage{enumerate, MnSymbol, bm}
\usepackage[nohug]{diagrams}
\usepackage{xcolor}

\newcommand{\GLtn}{\mathrm{GL}(2n,\mathbb{R})}
\newcommand{\Sptn}{\mathrm{Sp}(2n,\mathbb{R})}
\newcommand{\Otn}{\mathrm{O}(2n)}
\newcommand{\Mtn}{\mathrm{M}_{2n}(\mathbb{R})}
\newcommand{\Un}{\mathrm{U}(n)}

\newcommand{\sptn}{\mathfrak{sp}(2n,\mathbb{R})}
\newcommand{\otn}{\mathfrak{o}(2n)}
\newcommand{\gltn}{\mathfrak{gl}(2n,\mathbb{R})}

\newcommand{\JSp}{\operatorname{J}_{\mathrm{Sp}}}
\newcommand{\JO}{\operatorname{J}_{\mathrm{O}}}
\newcommand{\jSp}{\operatorname{j}_{\mathrm{Sp}}}
\newcommand{\jO}{\operatorname{j}_{\mathrm{O}}}
\newcommand{\cjSp}{\check{\operatorname{j}}_{\mathrm{Sp}}}

\newcommand{\Ad}{\mathrm{Ad}}
\newcommand{\ad}{\mathrm{ad}}
\newcommand{\omegacan}{\omega_{\mathbb{R}^{2n}}}
\newcommand{\Jomega}{\mathcal{J}(\mathbb{R}^{2n},\omegacan)}
\newcommand{\Grass}{\mathcal{G}_-(\mathbb{C}^{2n},\omegacan)}
\newcommand{\mcA}{\mathcal{A}}
\newcommand{\mcB}{\mathcal{B}}

\newtheorem{proposition}{Proposition}[section]
\newtheorem{corollary}[proposition]{Corollary}
\newtheorem{lemma}[proposition]{Lemma}

\theoremstyle{definition}
\newtheorem{definition}[proposition]{Definition}

\theoremstyle{remark}
\newtheorem*{remark}{Remark}

\begin{document}

\title{The frame bundle picture of Gaussian wave packet dynamics in semiclassical mechanics}
\author{Paul Skerritt$^1$}
\addtocounter{footnote}{1}
\footnotetext{Department of Mathematics, University of Surrey, Guildford GU2 7XH, United Kingdom.
\texttt{p.skerritt@surrey.ac.uk}}
\date{}

\maketitle

\begin{abstract}
Recently Ohsawa \cite{Ohsawa2015} has studied the Marsden-Weinstein-Meyer quotient of the manifold $T^*\mathbb{R}^n\times T^*\mathbb{R}^{2n^2}$ under a $\operatorname{O}(2n)$-symmetry, and has used this quotient to describe the relationship between two different parametrisations of Gaussian wave packet dynamics commonly used in semiclassical mechanics. In this paper we suggest a new interpretation of (a subset of) the unreduced space as being the frame bundle $\mathcal{F}(T^*\mathbb{R}^n)$ of $T^*\mathbb{R}^n$. We outline some advantages of this interpretation, and explain how it can be extended to more general symplectic manifolds using the notion of the diagonal lift of a symplectic form due to Cordero and de Le\'on \cite{CorderodeLeon1983}.
\end{abstract}

\section{Introduction}

\subsection{Motivation}
The Gaussian wave packet ansatz is frequently used in the study of the time-dependent Schr\"odinger equation, and its semiclassical limit. In particular, the wavefunction
\[
\psi(x) = \exp\left\lbrace \frac{i}{\hbar} \left[\frac{1}{2}(x-q)^\top (\mathcal{A}+i\mathcal{B})(x-q) + p^\top (x-q) + (\phi+i\delta) \right] \right\rbrace,
\]
parametrised by $(q,p)\in T^*\mathbb{R}^n$, $\phi,\delta\in\mathbb{R}$, and $\mathcal{A}+i\mathcal{B}\in\Sigma_n = \lbrace W\in \mathrm{M}_n(\mathbb{C})\,|\, W^\top =W, \mathrm{Im}\ W>0 \rbrace$, is well known to be an \emph{exact} solution of the Schr\"odinger equation for quadratic Hamiltonians, provided the parameters satisfy certain ODEs \cite{Heller1976}.

Faou and Lubich \cite{FaouLubich2006, Lubich2008} have shown that these parameter ODEs constitute a Hamiltonian system, and recently Ohsawa and Leok \cite{OhsawaLeok2013} have clarified the symplectic structure underlying these Hamiltonian dynamics. Their main observation is that $\Sigma_n$ is a symplectic manifold, the well-studied \emph{Siegel upper half plane} \cite{Siegel1973}. 

In \cite{Hagedorn1980}, Hagedorn introduced another parametrisation of Gaussian wavefunctions, by replacing $\mathcal{A}+i\mathcal{B}\in\Sigma_n$ by $PQ^{-1}$, where $Q, P\in\mathrm{M}_n(\mathbb{C})$ satisfying certain algebraic relations. In this parametrisation, the ODEs governing the parameter evolution are somewhat simpler. In \cite{Ohsawa2015}, Ohsawa explained how to interpret the relation between $Q,P$ and $\mathcal{A}+i\mathcal{B}$ as an instance of symplectic reduction on the symplectic manifold $T^*\mathbb{R}^n\times T^*\mathbb{R}^{n^2}$ under a certain right $\Otn$-action. The matrices $Q, P$ are seen to be coordinates on a level set of the momentum map corresponding to the $\Otn$-action, and symplectic reduction at this momentum level gives the Siegel upper half plane $\Sigma_n$. The dynamics in the unreduced and reduced spaces are then related by the quotient map in the usual way, explaining how the Hagedorn ODEs can also be interpreted as a Hamiltonian system.

\subsection{Main Results and Outline}
The purpose of the present paper is to reinterpret Ohsawa's results, and to suggest some possible extensions. Specifically, we interpret an open subset of Ohsawa's unreduced space $T^*\mathbb{R}^n\times T^*\mathbb{R}^{2n^2}$ as the \emph{frame bundle} $\mathcal{F}(T^*\mathbb{R}^n)\simeq T^*\mathbb{R}^n \times \GLtn$ of $T^*\mathbb{R}^n$, and give an intrinsic description of the symplectic structure on $\mathcal{F}(T^*\mathbb{R}^n)$. There are several advantages to this picture:
\begin{itemize}
\item there exists a dual pair structure on $\GLtn$, generated by left multiplication by $\Sptn$ and right multiplication by $\Otn$ . The existence of this dual pair allows one to realise the symplectic reduced space of $\mathcal{F}(T^*\mathbb{R}^n)$ at \emph{all} values of $\Otn$-momentum in terms of adjoint orbits of $\sptn$. 
\footnote{Although we consider one particular adjoint orbit in this paper, it should be possible to associate semiclassical wavefunctions with arbitrary integral adjoint orbits. We plan to explore this possibility in future work.} 
\item
the physical picture of moving frames on $T^*\mathbb{R}^n$ is intuitive, and gives insight into the definition of the Hamiltonian on the unreduced space. It also has overlap with other work on semiclassical mechanics, in particular the \emph{nearby orbit approximation} of Littlejohn \cite{Littlejohn1986}.
\item
the frame bundle picture can be extended to more general symplectic manifolds than $T^*\mathbb{R}^n$, using the results of Cordero et al. \cite{CorderodeLeon1983} on the lifting of symplectic structures to frame bundles.
\end{itemize}

We now outline the stucture of the paper. In Section \ref{sec:sympMn}, we describe a symplectic structure on the set $\Mtn$ of $2n\times 2n$ dimensional real matrices, demonstrate that the obvious left $\Sptn$- and right $\Otn$-actions are Hamiltonian, and calculate the corresponding momentum maps, and their Lie algebra-valued counterparts. In Section \ref{sec:dualpairs}, we review some properties of dual pairs that will be needed in the sequel. In Section \ref{sec:dualGLtn}, we restrict our symplectic manifold to $\GLtn\subset \Mtn$, and demonstrate that the two group actions give a dual pair structure on this restricted space. In Section \ref{sec:reduction}, we discuss reduction of $\GLtn$ under the $\Otn$-action at a particular value of momentum, use the results of Section \ref{sec:dualpairs} to describe the reduced space as an adjoint orbit in $\sptn$, and give a geometric description of this adjoint orbit as the set of $\omegacan$-compatible complex structures, where $\omegacan$ is the canonical symplectic structure on $\mathbb{R}^{2n}$. In Section \ref{sec:UHPcoords}, we use the standard isomorphism of $\omegacan$-compatible complex structures with the positive Lagrangian Grassmannian to introduce coordinates on the adjoint orbit, and demonstrate that in these coordinates the symplectic projection and reduced symplectic form agree with those obtained by Ohsawa. In Section \ref{sec:dynamics}, starting with the cotangent symplectic structure $\omega=\sum_{\alpha=1}^n dq^\alpha\wedge dp_\alpha$ and an arbitrary Hamiltonian $H$ on $T^*\mathbb{R}^n$, we describe the $\hbar$-diagonal symplectic structure $\Omega^{\hbar}$ and $\hbar$-lifted Hamiltonian $H^{\hbar}$ on $\mathcal{F}(T^*\mathbb{R}^n)$, which agree with those defined by Ohsawa on their common domain. We give a geometric interpretation of $H^{\hbar}$, demonstrate that for quadratic $H$ one obtains the dynamics of Hagedorn, and emphasise that the $\Otn$-invariance of $H^{\hbar}$ leads to conservation of \emph{symplectic} frame. We also describe the dynamics on the reduced space, and show it reproduces those of Heller. Finally in Section \ref{sec:deLeon}, we explain how the results of Cordero et al. \cite{CorderodeLeon1983} can be used to extend the construction of a symplectic structure on the frame bundle to arbitrary symplectic manifolds.

\section{The symplectic geometry of \texorpdfstring{$\Mtn$}{M(2n,R)}}\label{sec:sympMn}

In this section, we outline the symplectic structure on $\Mtn$, describe two natural symplectic group actions on $\Mtn$, and calculate the momentum maps associated with these actions. 

\subsection{The symplectic form on \texorpdfstring{$\Mtn$}{M(2n,R)}}
Take $\mathbb{R}^{2n}$ with its canonical symplectic structure
\[
\omegacan(v,w) := v^\top\mathbb{J} w,
\]
where $\mathbb{J} := \begin{bmatrix} 0_n & I_n \\ -I_n & 0_n \end{bmatrix}$. The $2n$-fold direct sum of $(\mathbb{R}^{2n}, \omegacan)$ is naturally isomorphic to the space $\mathrm{M}_{2n}(\mathbb{R})$ of real $2n\times2n$ matrices with the symplectic form
\[
\Omega(E,F) := \sum_{a=1}^{2n}\omegacan(E_a,F_a) = \operatorname{Tr}(E^\top\mathbb{J}F),
\]
where $E_a$ denotes the $a$th column of $E \in\mathrm{M}_{2n}(\mathbb{R})$, considered as a vector in $\mathbb{R}^{2n}$. We can consider $(\Omega, \Mtn)$ as a symplectic \emph{manifold}, using the canonical isomorphism $T_E\Mtn \simeq \Mtn$, and we denote the induced symplectic form by $\Omega$ also. Explicitly, if $V_E = \frac{d}{dt}\big\vert_{t=0}(E+tV)$ for $V\in\Mtn$ etc., then
\[
\Omega_E(V_E,W_E) := \Omega(V,W) = \operatorname{Tr}(V^\top\mathbb{J}W).
\]

\subsection{Commuting symplectic actions}
Let $L_S:\Mtn\rightarrow\Mtn$ denote left multiplication by $S\in\Sptn$, and $R_O:\Mtn\rightarrow\Mtn$ denote right multiplication by $O\in\Otn$. Using the expressions for the differentials of these actions
\[
T_EL_S(V_E) = (SV)_{SE}\qquad\textrm{and}\qquad T_ER_O(V_E) = (VO)_{EO}
\]
it is straightforward to check that $L$ and $R$ define left and right symplectic actions respectively, i.e., 
for any $E\in\Mtn$,
\[
(L_S^*\Omega)_E = \Omega_E \qquad\textrm{and}\qquad (R_O^*\Omega)_E = \Omega_E.
\]
The momentum maps corresponding to these two symplectic actions have a standard form.
\begin{proposition}
\begin{enumerate}[(i)]
\item
The momentum map $\JSp:\Mtn\rightarrow \sptn^*$ corresponding to the left $\Sptn$-action is given by
\[
\langle \JSp(E), \zeta\rangle = \frac{1}{2}\Omega(\zeta E,E).
\]
\item
The momentum map $\JO:\Mtn\rightarrow \otn^*$ corresponding to the right $\Otn$-action is given by
\[
\langle \JO(E),\xi \rangle = \frac{1}{2}\Omega(E\xi,E).
\]
\end{enumerate}
\end{proposition}

\begin{proof}
Both results follow from the general expression for the momentum map of a linear symplectic action on a symplectic vector space\textemdash see for example \cite[Section 12.4, Example (a)]{MarsdenRatiu1999}.
\end{proof}

\begin{remark} Both momentum maps are easily seen to be equivariant,
\[
\JSp(SE) = \Ad_{S^{-1}}^*(\JSp(E)) \qquad\textrm{and}\qquad \JO(EO) = \Ad_O^*(\JO(E)),
\]
(again, a general result for linear symplectic actions).
\end{remark}

\subsection{Lie algebra-valued momentum maps}
For any (real) Lie subalgebra $\mathfrak{g}\subset\mathfrak{gl}(N,\mathbb{C})$, consider the trace form $\llangle\cdot,\cdot\rrangle:\mathfrak{g}\times\mathfrak{g}\rightarrow\mathbb{C}$ defined by
\[
\llangle \xi,\zeta\rrangle := \frac{1}{2}\operatorname{Tr}(\xi\zeta).
\]
If $\mathfrak{g}$ is invariant under conjugate transpose, then $\llangle\cdot,\cdot\rrangle$ is non-degenerate, as is its real part, since 
\[
\llangle \xi,\xi^\dagger\rrangle = \operatorname{Re}\llangle \xi,\xi^\dagger\rrangle>0\textrm{ for }\xi\ne 0.
\]
In particular, it is non-degenerate and real-valued on $\sptn$ and $\otn$. Using the trace form to identify Lie algebras with their duals, we can write down Lie algebra-valued versions of the momentum maps discussed in the previous section.

\begin{proposition}\phantomsection\label{prop:lamom}
\begin{enumerate}[(i)]
\item
The $\sptn$-valued momentum map $\jSp:\Mtn\rightarrow \sptn$ corresponding to $\JSp$ is given by
\[
\jSp(E) = -EE^\top\mathbb{J}.
\]
\item
The $\otn$-valued momentum map $\jO:\Mtn\rightarrow\otn$ corresponding to $\JO$ is given by
\[
\jO(E) = -E^\top\mathbb{J}E.
\]
\end{enumerate}
\end{proposition}

\begin{proof}
\begin{enumerate}[(i)]
\item
Denoting the natural pairing of $\sptn^*$ with $\sptn$ by $\langle\cdot,\cdot\rangle$, we have that for $\zeta\in\sptn$,
\begin{align*}
\langle \JSp(E),\zeta\rangle &= \frac{1}{2}\Omega(\zeta E,E) = \frac{1}{2}\operatorname{Tr}(E^\top\zeta^\top\mathbb{J} E) \\
&= \frac{1}{2}\operatorname{Tr}(EE^\top\zeta^\top\mathbb{J})
= -\frac{1}{2}\operatorname{Tr}(EE^\top\mathbb{J}\zeta) \qquad\textrm{since }\zeta\in\sptn\\
&= \llangle -EE^\top\mathbb{J},\zeta\rrangle.
\end{align*}
Also
\[
(-EE^\top\mathbb{J})^\top\mathbb{J} + \mathbb{J}(-EE^\top\mathbb{J}) = \mathbb{J}EE^\top\mathbb{J} - \mathbb{J}EE^\top\mathbb{J} = 0.
\]
So $-EE^\top\mathbb{J}\in\sptn$, and the result follows.
\item
Similar.
\end{enumerate}
\end{proof}

\begin{remark}
Again, the Lie algebra-valued momentum maps are equivariant
\[
\jSp(SE) = \Ad_S(\jSp(E)) \qquad\textrm{and}\qquad \jO(EO) = \Ad_{O^{-1}}(\jO(E)).
\]
\end{remark}

\section{Mutually transitive actions and dual pairs}\label{sec:dualpairs}

In Section \ref{sec:dualGLtn} we will demonstrate that the left $\Sptn$-action and right $\Otn$-action define a dual pair structure on a suitable subset of $\Mtn$. 
In anticipation, we here introduce the notion of \emph{mutually transitive actions}. We then explain how mutual transitivity allow us to view reduced spaces of one action as coadjoint orbits of the other. Finally, we indicate the relation of mutual transitivity to the more standard notion of a dual pair.

A fuller treatment of dual pairs and related concepts can be found in \cite[Section IV.7]{LibermannMarle1987}, \cite[Chapter 11]{OrtegaRatiu2004}, and \cite{BalleierWurzbacher2012}. We also refer to \cite{SkerrittVizman2019} for a more thorough discussion of mutual transitivity.

\subsection{Mutually transitive actions}

Let $(N,\Omega)$ be a symplectic manifold, and let $\Phi_1:G_1\times N \rightarrow N$ and $\Phi_2:G_2\times N\rightarrow N$ be symplectic actions. We assume $N$, $G_1$, and $G_2$ are all finite-dimensional.

\begin{definition}\label{defn:mt}
We say the actions $\Phi_1, \Phi_2$ are \emph{mutually transitive} if the following three properties hold:
\begin{itemize}
\item
$\Phi_1$ and $\Phi_2$ commute,
\item
$\Phi_1$ and $\Phi_2$ are Hamiltonian actions, with corresponding equivariant momentum maps $\operatorname{J}_1$ and $\operatorname{J}_2$,
\item
each level set of $\operatorname{J}_1$ is a $G_2$-orbit and vice versa, i.e., for any $x\in N$,
\[
\operatorname{J}_1^{-1}(\operatorname{J}_1(x)) = G_2\cdot x \qquad\textrm{and}\qquad \operatorname{J}_2^{-1}(\operatorname{J}_2(x)) = G_1\cdot x.
\]
\end{itemize}
\end{definition}

Denoting the coadjoint orbit in $\mathfrak{g}_i^*$ through $\mu_i$ by $\mathcal{O}_{\mu_i}$, we then have

\begin{proposition}
Let $\Phi_1, \Phi_2$ be mutually transitive actions.Then for all $x\in N$,
\[
\operatorname{J}_1^{-1}(\mathcal{O}_{\operatorname{J}_1(x)}) = \operatorname{J}_2^{-1}(\mathcal{O}_{\operatorname{J}_2(x)}).
\]
\end{proposition}

\begin{proof}
\begin{align*}
\mathrm{J}_1^{-1}(\mathcal{O}_{\operatorname{J}_1(x)}) &= G_1\cdot \operatorname{J}_1^{-1}(\operatorname{J}_1(x)) \qquad\textrm{since $\operatorname{J}_1$ is $G_1$-equivariant} \\
&= G_1\cdot (G_2\cdot x) \qquad\textrm{since $\Phi_2$ is transitive on the fibres of $\operatorname{J}_1$}\\
&= G_2\cdot (G_1\cdot x) \qquad \textrm{since the actions $\Phi_1$ and $\Phi_2$ commute}\\
&= G_2\cdot\operatorname{J}_2^{-1}(\operatorname{J}_2(x)) \qquad\textrm{since $\Phi_1$ is transitive on the fibres of $\operatorname{J}_2$}\\
&= \operatorname{J}_2^{-1}(\mathcal{O}_{\operatorname{J}_2(x)}) \qquad \textrm{since $\operatorname{J}_2$ is $G_2$-equivariant}.
\end{align*}
\end{proof}

\begin{corollary}[{\cite[Theorem 2.8(i)]{BalleierWurzbacher2012}}]\label{coro}
There exists a one-to-one correspondence between coadjoint orbits in $\operatorname{J}_1(N)$ and $\operatorname{J}_2(N)$ given by
\[
\mathcal{O}_{\mu_1}\mapsto \operatorname{J}_2(\operatorname{J}_1^{-1}(\mathcal{O}_{\mu_1})) { =\operatorname{J}_2(\operatorname{J}_1^{-1}(\mu_1))}
\]
or equivalently
\[
\qquad\qquad\mathcal{O}_{\operatorname{J}_1(x)}\mapsto \mathcal{O}_{\operatorname{J}_2(x)} \qquad\textrm{for }x\in N.
\]
\end{corollary}

\subsection{The relation between coadjoint orbits and reduced spaces}

From now on we assume that all group actions $\Phi_i$ are both free and proper (see \cite{SkerrittVizman2019} for a discussion when these conditions do not hold). 

For $\mu_1\in\mathfrak{g}_1^*$, let $(G_1)_{\mu_1}$ denote the coadjoint stabiliser of $\mu_1$, and $N_{\mu_1}$ the quotient space $\operatorname{J}_1^{-1}(\mu_1)/(G_1)_{\mu_1}$ at $\mu_1$. It is well known that $N_{\mu_1}$ may be given a smooth structure making $\sigma_1:\operatorname{J}_1^{-1}(\mu_1)\rightarrow N_{\mu_1}$ a submersion, and a symplectic structure $\Omega_{N_{\mu_1}}$ satisfying 
\begin{equation}\label{eqn:redform}
i_{\mu_1}^*\Omega = \sigma_1^*\Omega_{N_{\mu_1}},
\end{equation}
where $i_{\mu_1}:\operatorname{J}_1^{-1}(\mu_1)\rightarrow N$ is the inclusion (see for example \cite[Theorem 6.1.1 (i)]{OrtegaRatiu2004}). The pair $(N_{\mu_1}, \Omega_{N_{\mu_1}})$ is called the \emph{reduced space} or \emph{Marsden-Weinstein-Meyer quotient}. 

For mutually transitive actions, the reduced space under one action is symplectomorphic to a coadjoint orbit of the other action. To prove this result, we first give the following lemma:

\begin{lemma}\label{lem:momsub}
The smooth map $\operatorname{J}_2:N\rightarrow \mathfrak{g}_2^*$ restricts to a smooth submersion $\operatorname{J}_2:\operatorname{J}_1^{-1}(\mu_1)\rightarrow \mathcal{O}_{\mu_2}$ satisfying
\begin{equation}\label{eqn:KKSform}
i_{\mu_1}^*\Omega = \operatorname{J}_2^*\Omega^+_{\mathcal{O}_{\mu_2}},
\end{equation}
where $\mathcal{O}_{\mu_2}\subset\mathfrak{g}_2^*$ is a coadjoint orbit, and $\Omega^+_{\mathcal{O}_{\mu_2}}$ is the positive KKS symplectic form on $\mathcal{O}_{\mu_2}$.
\end{lemma}

\begin{proof}
Since $\Phi_1$ is free, the momentum map $\operatorname{J}_1$ is a submersion \cite[Corollary 4.5.13]{OrtegaRatiu2004}, and hence the level set $\operatorname{J}_1^{-1}(\mu_1)$ is an embedded submanifold of $N$ \cite[Section 1.1.13]{OrtegaRatiu2004}. Since $\operatorname{J}_1^{-1}(\mu_1)$ is a  $G_2$-orbit, and $\operatorname{J}_2$ is $G_2$-equivariant, the image $\operatorname{J}_2(\operatorname{J}_1^{-1}(\mu_1))$ equals some coadjoint orbit $\mathcal{O}_{\mu_2}\subset \mathfrak{g}_2^*$, which is an initial submanifold \cite[Proposition 2.3.12 (i)]{OrtegaRatiu2004}. Hence $\operatorname{J}_2:N\rightarrow \mathfrak{g}_2^*$ restricts to a smooth map $\operatorname{J}_2:\operatorname{J}_1^{-1}(\mu_1)\rightarrow \mathcal{O}_{\mu_2}$, and by $G_2$-equivariance this restriction is a submersion.

Taking $\xi,\zeta\in\mathfrak{g}_2$, $x\in\operatorname{J}_1^{-1}(\mu_1)$, and using the $G_2$-equivariance of $\operatorname{J}_2$, we have
\begin{align*}
\Omega_x(\xi\cdot x, \zeta\cdot x) &= d_x\langle \operatorname{J}_2,\xi\rangle (\zeta\cdot x) = \zeta\cdot x\langle \operatorname{J}_2,\xi\rangle = \langle -\ad_\zeta^*\operatorname{J}_2(x),\xi\rangle \\
&= \langle \operatorname{J}_2(x),[\xi,\zeta]\rangle = (\Omega^+_{\mathcal{O}_{\mu_2}})_{\operatorname{J}_2(x)}(-\ad_\xi^*\operatorname{J}_2(x),-\ad_\zeta^*\operatorname{J}_2(x)).
\end{align*}
Again using equivariance of $\operatorname{J}_2$, plus the fact that $\operatorname{J}_1^{-1}(\mu_1)$ is a $G_2$-orbit, gives \eqref{eqn:KKSform}.
\end{proof}

\begin{proposition}[{\cite[Theorem 2.8(iii)]{BalleierWurzbacher2012}}]\label{prop:redcoad}
Let $\Phi_1, \Phi_2$ be mutually transitive actions on $N$. Then any reduced space under the $G_1$-action is symplectomorphic to a coadjoint orbit in $\operatorname{J}_2(N)\subset\mathfrak{g}_2^*$, and similarly with 1 and 2 interchanged. Explicitly, for $x\in N$,
\[
N_{\operatorname{J}_1(x)} \simeq \mathcal{O}_{\operatorname{J}_2(x)}, \qquad N_{\operatorname{J}_2(x)} \simeq \mathcal{O}_{\operatorname{J}_1(x)},
\]
via a resp. $G_2$- and $G_1$-equivariant symplectomorphism.
\end{proposition}

\begin{proof} 
Choose $x\in N$, and let $\mu_i=\operatorname{J}_i(x),\,i=1,2$. By equivariance of $\operatorname{J}_1$ and $\operatorname{J}_2$ and the mutual transitivity property, it is not difficult to show that for any $y\in \operatorname{J}_1^{-1}(\mu_1)$, 
\[
(G_1)_{\operatorname{J}_1(y)}\cdot y = (G_2)_{\operatorname{J}_2(y)}\cdot y = G_1\cdot y \cap G_2\cdot y.
\] 
Hence the fibres of $\sigma_1:\operatorname{J}_1^{-1}(\mu_1)\rightarrow N_{\mu_1}$ and $\operatorname{J}_2:\operatorname{J}_1^{-1}(\mu_1)\rightarrow \mathcal{O}_{\mu_2}$ agree, and we get a bijection $\check{\operatorname{J}}_2:N_{\mu_1}\rightarrow \mathcal{O}_{\mu_2}$ making the following diagram commute
\begin{equation}\label{diag:redcoad}
\begin{diagram}
&&\operatorname{J}_1^{-1}(\mu_1) \\
&\ldTo^{\sigma_1} &&\rdTo^{\operatorname{J}_2}\\
N_{\mu_1} &&\rTo^{\check{\operatorname{J}}_2}&& \mathcal{O}_{\mu_2}
\end{diagram}
\end{equation}

Since $\sigma_1$ and $\operatorname{J}_2$ are submersions, $\check{\operatorname{J}}_2$ is a diffeomorphism. Equations \eqref{eqn:redform} and \eqref{eqn:KKSform}, and commutativity of \eqref{diag:redcoad}, implies that
\[
\sigma_1^*\Omega_{N_{\mu_1}} = \operatorname{J}_2^*\Omega_{\mathcal{O}^+_{\mu_2}} = \sigma_1^*\check{\operatorname{J}}_2^*\Omega_{\mathcal{O}^+_{\mu_2}}.
\]
Since $\sigma_1$ is a submersion, it follows that 
\[
\Omega_{N_{\mu_1}} = \check{\operatorname{J}}_2^*\Omega^+_{\mathcal{O}_{\mu_2}},
\] 
i.e., $\check{\operatorname{J}}_2$ is a symplectomorphism.

\end{proof}

\begin{remark}
The map $\check{\operatorname{J}}_2$ has a simple interpretation: since the $G_1$- and $G_2$-actions commute, the $G_2$-action drops to the reduced space $N_{\operatorname{J}_1(x)}$. Using $(i^{\mu_1})^*\Omega = (\sigma_1)^*\Omega_{N_{\mu_1}}$, we see that $\check{\operatorname{J}}_2$ is simply the momentum map for this reduced action, and is itself equivariant.
\end{remark}

\subsection{The relation to dual pairs}
In this subsection we make contact with the notion of dual pair (in the sense of Weinstein \cite{Weinstein1983}).

\begin{definition}[{\cite[Appendix E]{Blaom2001}}]
Let $N$ be a symplectic manifold, and $P_1, P_2$ Poisson manifolds. A pair of Poisson maps 
\[
P_1 \stackrel{\operatorname{J}_1}{\longleftarrow} N \stackrel{\operatorname{J}_2}{\longrightarrow} P_2
\]
is called a \emph{full dual pair} if $\operatorname{J}_1, \operatorname{J}_2$ are submersions, and
\[
(\ker T\operatorname{J}_1)^\Omega = \ker T\operatorname{J}_2.
\]
\end{definition}

\begin{proposition}
If $\Phi_1, \Phi_2$ are mutually transitive and free, then their momentum maps $\operatorname{J}_1, \operatorname{J}_2$ form a full dual pair.
\end{proposition}

\begin{proof}
Let $x\in N$. The identity $G_1\cdot x = \operatorname{J}_2^{-1}(\operatorname{J}_2(x))$ implies that
\[
\mathfrak{g}_1\cdot x = T_x(\operatorname{J}_2^{-1}(\operatorname{J}_2(x))).
\]
A standard result says that the former equals $(\ker T_x\operatorname{J}_1)^\Omega$.
Also, as in Lemma \ref{lem:momsub}, freeness of $\Phi_1, \Phi_2$ implies that $\operatorname{J}_1, \operatorname{J}_2$ are submersions. Then $\operatorname{J}_2^{-1}(\operatorname{J}_2(x))$ is an embedded submanifold, and so 
\[
T_x(\operatorname{J}_2^{-1}(\operatorname{J}_2(x))) = \ker T_x\operatorname{J}_2.
\]
The result follows.
\end{proof}

\section{The dual pair structure on \texorpdfstring{$\GLtn\subset \Mtn$}{GL(2n,R)}}\label{sec:dualGLtn}

The subset $\GLtn$ of $\Mtn$ is open, and hence $\Omega$ restricts to an non-degenerate form on $\GLtn$. Denote the restriction of $\Omega$ and the momentum maps to $\GLtn$ by the same symbols for simplicity. 

With this restriction, we can demonstrate that the left $\Sptn$- and right $\Otn$-actions are mutually transititive (Definition \ref{defn:mt}). For this, it is more convenient to work with the equivalent double fibration
\begin{diagram}
&&\GLtn \\
&\ldTo^{\jSp} && \rdTo^{\jO}\\
\sptn &&&& \otn 
\end{diagram}

\begin{proposition}
\begin{enumerate}[(i)]
\item
The left $\Sptn$-action acts transitively on the fibres of $\jO$.
\item
The right $\Otn$-action acts transitively on the fibres of $\jSp$.
\end{enumerate}
\end{proposition}

\begin{proof}
\begin{enumerate}[(i)]
\item
Using Proposition \ref{prop:lamom}, we have that for any $E, E'\in\GLtn$,
\begin{align*}
&\jO(E') = \jO(E) \\
\iff  &-(E')^\top\mathbb{J} E' = -E^\top\mathbb{J} E \\
\iff &(E'E^{-1})^\top\mathbb{J} (E'E^{-1}) = \mathbb{J}\\
\iff &E' = SE \qquad\textrm{for some }S\in\Sptn.
\end{align*}
\item
Similar.
\end{enumerate}
\end{proof}

\begin{proposition}
The left $\Sptn$-action and right $\Otn$-action are free and proper.
\end{proposition}

\begin{proof}
We prove the result for the $\Sptn$-action only; the $\Otn$ case is similar.
Freeness is trivial. To prove properness: suppose we have two convergent sequences $(E_k)$ and $(S_k\cdot E_k)$ in $\GLtn$ with $E_k\rightarrow E$ and $S_k E_k\rightarrow E$. Then $S_k = S_k E_k E_k^{-1} \rightarrow E E^{-1} = I$.
\end{proof}

\begin{remark}
For any $E\in\Mtn$, it remains true that $E\cdot \Otn \subset \jSp^{-1}(\jSp(E))$ and $\Sptn\cdot E \subset \jO^{-1}(\jO(E))$. However restriction to $\GLtn\subset\Mtn$ is necessary to get equality in the second inclusion. For example
\[
\jO^{-1}(0) = \lbrace E\in\Mtn\,|\, \mathrm{span}\,E \textrm{ is an isotropic subspace of }\mathbb{R}^{2n}\rbrace,
\]
and this is not simply an orbit of $\Sptn$ (since for example the $\Sptn$-action preserves the dimension of $\mathrm{span}\,E$). It is however possible to prove that $E\cdot\Otn = \jSp^{-1}(\jSp(E))$ for \emph{all} $E\in\Mtn$.
\end{remark}

With the restriction from $\Mtn$ to $\GLtn$, the momentum map $\jSp$ acquires a nice interpretation.
As before, let $E_a$ denote the $a$th column of $E$, so $\lbrace E_1,\ldots E_{2n}\rbrace$ defines a basis of $\mathbb{R}^{2n}$. Let $g(E)$ denote a metric on $\mathbb{R}^{2n}$ with respect to which the basis $\lbrace E_a\rbrace$ is orthonormal. Then any basis related to $E$ by the right $\Otn$-action defines the same metric, i.e.,
\[
g(E\cdot O) = g(E) \qquad\textrm{for all }O\in\Otn,
\]
and in fact the set of metrics on $\mathbb{R}^{2n}$ is in 1-1 correspondence with the $\Otn$-orbits of $\GLtn$. We then have the following result:

\begin{proposition}
The metric $g(E)$ and canonical symplectic form $\omegacan$ are related by
\[
\omegacan(u,v) = g(E)(\jSp(E)u,v).
\]
\end{proposition}

\begin{proof}
If $\lbrace s_a \rbrace$ is the standard basis of $\mathbb{R}^{2n}$ (so $E_a = \sum_{b=1}^{2n}E_{ba}s_b$), the definition of $g(E)$ implies that
\[
g(E)(s_a,s_b) = [(E E^\top)^{-1}]_{ab}.
\]
From this, it follows that
\begin{align*}
g(E)(\jSp(E)u,v) &= (\jSp(E)u)^\top (E E^\top)^{-1} v = u^\top (-E E^\top\mathbb{J})^\top (E E^\top)^{-1}v \\
&= u^\top \mathbb{J} v = \omegacan(u,v).
\end{align*}
\end{proof}

Inverting this identity to give
\[
g(E)(u,v) = \omegacan((\jSp(E))^{-1}u,v),
\]
we see that $\jSp(E)$ defines the metric $g(E)$, and so essentially corresponds to it. In the next section we will consider this correspondence in the particular case that the ordered basis $(E_a)$ forms a \emph{symplectic} frame of $(\mathbb{R}^{2n}, \omegacan)$.

\section{Reduction}\label{sec:reduction}

Following Ohsawa \cite{Ohsawa2015}, we consider in this section the Marsden-Weinstein quotient for the right $\Otn$-action at a particular value of momentum. Using the correspondence between reduced spaces and coadjoint orbits provided by the dual pair structure, we will give a geometrical characterisation of this space. Using this characterisation, we will explain how to naturally introduce a global coordinate chart on the reduced space, reproducing Ohsawa's description of the space as the Siegel upper half plane.

\subsection{Reduction through a general point}
It will be more convenient to work with the Lie-algebra valued momentum maps $\jO$ and $\jSp$. For $\xi\in\mathfrak{g}$, now let $G_\xi$ denote the \emph{adjoint} stabiliser of $\xi$, and $\mathcal{O}_\xi\subset \mathfrak{g}$ the \emph{adjoint} orbit through $\xi$. Writing $N=\GLtn$, the reduced space through $E\in \GLtn$ is
\[
N_{\jO(E)} = \frac{\jO^{-1}(\jO(E))} {\Otn_{\jO(E)}}.
\]
Using the trace form to transfer the Poisson structure from $\sptn^*$ to
$\sptn$, Proposition \ref{prop:redcoad} tells use there exists a $\Sptn$-equivariant symplectomorphism
\[
\cjSp: N_{\jO(E)} \rightarrow \mathcal{O}_{\jSp(E)}.
\]

\subsection{Reduction through the identity}
We consider the reduced space for the left $\Otn$-action through $I=I_{2n}\in\GLtn$. From Proposition \ref{prop:lamom}(ii) we see that the level set of $\jO$ through $I$ is the set
\[
\jO^{-1}(\jO(I)) = \lbrace E\in\GLtn\,|\, -E^\top\mathbb{J} E = -\mathbb{J} \rbrace = \Sptn,
\]
which corresponds to the set of \emph{symplectic} frames of $(\mathbb{R}^{2n},\omegacan)$. Meanwhile,
\[
\Otn_{\jO(I)} = \lbrace O\in\Otn\,|\, \Ad_O(-\mathbb{J}) = -\mathbb{J} \rbrace = \Sptn\cap \Otn \simeq \Un,
\]
and we see that the reduced space $N_{\jO(I)}$ is isomorphic to $\Sptn/\Un$, as originally obtained by Ohsawa \cite{Ohsawa2015}.

\subsection{Geometric interpretation of the adjoint orbit \texorpdfstring{$\mathcal{O}_{\jSp(I)}$}{}}
As explained above, there is a $\Sptn$-equivariant map
\[
\cjSp:N_{\jO(I)} \rightarrow \mathcal{O}_{\jSp(I)}.
\]
We now give an alternative geometric characterisation of $\mathcal{O}_{\jSp(I)}$ in terms of complex structures on $\mathbb{R}^{2n}$.

\begin{definition}
Let $j:\mathbb{R}^{2n}\rightarrow\mathbb{R}^{2n}$ be a complex structure on $\mathbb{R}^{2n}$, i.e., an $\mathbb{R}$-linear map satisfying $j^2=-I$. The map $j$ is called \emph{$\omegacan$-compatible} if it satisfies the two properties
\begin{enumerate}[(i)]
\item
$j$ is symplectic (i.e., $j\in\Sptn$), and
\item
the bilinear form $g_j(u,v):= \omegacan(u,jv)$ is positive-definite.
\end{enumerate}
Denote the space of $\omegacan$-compatible complex structures on $\mathbb{R}^{2n}$ by $\Jomega$.
\end{definition}

\begin{proposition}
The adjoint orbit $\mathcal{O}_{\jSp(I)}\subset \sptn$ is precisely the set $\Jomega$.
\end{proposition}

\begin{proof}
An arbitrary element of $\zeta\in\mathcal{O}_{\jSp(I)}$ is of the form $\zeta = \Ad_S[\jSp(I)] = \Ad_S[-\mathbb{J}]$. Since $(-\mathbb{J})^2 = -I$ we have $\zeta^2 = -I$, and since $-\mathbb{J}\in\Sptn$ we have $\zeta = \Ad_S(-\mathbb{J}) \in\Sptn$. Finally, writing $\zeta = -SS^\top\mathbb{J}$, we see that for $u\in\mathbb{R}^{2n}-\lbrace 0\rbrace$,
\[
g_\zeta(u,u) = \omegacan(u,\zeta u) = u^\top\mathbb{J} (-SS^\top\mathbb{J})u = u^\top(S^\top\mathbb{J})^\top(S^\top\mathbb{J})u>0.
\]
So we have that $\mathcal{O}_{\jSp(I_{2n})}\subset \Jomega$.

Conversely, suppose $j\in\Jomega$. Then
\[
g_j(ju,jv) = \omegacan(ju,j^2v) = \omegacan(u,jv) = g_j(u,v),
\]
the second identity following since $j$ is symplectic. Hence $j$ is orthogonal with respect to the metric $g_j$. A standard theorem guarantees the existence of a real orthonormal basis $\lbrace u_\alpha, v_\alpha\,|\,\alpha=1,\ldots, n\rbrace$ of $\mathbb{R}^{2n}$ such that $j(u_\alpha - iv_\alpha)  = e^{i\theta_\alpha}(u_\alpha - iv_\alpha)$. The identity $j^2 = -I$ implies that $e^{i\theta_\alpha} = \pm i$ for all $\alpha$, and by switching $u_\alpha$ and $v_\alpha$ if necessary we can choose $e^{i\theta_\alpha} = i$ for all $\alpha$. So
\[
j u_\alpha = v_\alpha, \qquad jv_\alpha = -u_\alpha.
\]
It can be verified that $(u_1,\ldots, u_n,v_1,\ldots, v_n)$ forms a symplectic basis for $(\mathbb{R}^{2n},\omegacan)$. For example
\[
\omegacan(u_\alpha, v_\beta) = \omegacan(u_\alpha, ju_\beta) = g_j(u_\alpha, u_\beta) = \delta_{\alpha\beta}
\]
by orthonormality. Taking $E = \begin{bmatrix}u_1 &\ldots& u_n & v_1 &\ldots &v_n \end{bmatrix}\in\Sptn$, the above relations can be written $jE = -E\mathbb{J}$, which implies
\[
j = -E\mathbb{J} E^{-1} = \Ad_E[-\mathbb{J}] =\Ad_E[\jSp(I)].
\]
Thus $\Jomega\subset\mathcal{O}_{\jSp(I)}$.
\end{proof}

Adapting diagram (\ref{diag:redcoad}) to our case, we obtain
\begin{equation}\label{diag:Spfirst}
\begin{diagram}
&&\jO^{-1}(\jO(I)) = \Sptn \\
&\ldTo^{\sigma_\mathrm{O}} && \rdTo^{\jSp}\\
N_{\jO(I)} && \rTo^{\cjSp} && \mathcal{O}_{\jSp(I)} = \Jomega
\end{diagram}
\end{equation}
with $\cjSp$ an $\Sptn$-equivariant symplectomorphism.
An equivalent diagram was recently obtained independently by Ohsawa and Tronci \cite{OhsawaTronci2017}.

\section{The upper half plane coordinates}\label{sec:UHPcoords}
As demonstrated in the previous section, the adjoint orbit $\mathcal{O}_{\jSp(I)}$ can be naturally viewed as the space of $\omegacan$-compatible complex structures $\Jomega$. It is a standard result that $\Jomega$ has another natural interpretation, that of the negative Lagrangian Grassmannian $\Grass$. In this section, we outline this correspondence, and explain how it can be used to introduce a coordinate system on $\mathcal{O}_{\jSp(I)}$.

\subsection{The negative Lagrangian Grassmannian}
Consider the complexification $(\mathbb{R}^{2n})\otimes_{\mathbb{R}}\mathbb{C} \simeq \mathbb{C}^{2n}$ of $\mathbb{R}^{2n}$. Real-linear functions on $\mathbb{R}^{2n}$ extend to complex-linear functions on $\mathbb{C}^{2n}$ is the obvious way, and we will generally use the same symbol to denote both a function and its complex extension.

Define a sesquilinear form $s:\mathbb{C}^{2n}\times\mathbb{C}^{2n}\rightarrow \mathbb{C}$ by
\[
s(w,z) := -i\omegacan(w,\overline{z}).
\]

\begin{definition}
The \emph{negative Lagrangian Grassmannian} $\Grass$ is the set of (complex) Lagrangian subspaces of $\mathbb{C}^{2n}$ on which $s$ restricts to a negative-definite form,
\[
\Grass := \lbrace \Gamma\textrm{ a subspace of }\mathbb{C}^{2n}\,|\,\Gamma\textrm{ is Lagrangian}, s\vert_{\Gamma}<0 \rbrace.
\]
\end{definition}

\begin{proposition}[{\cite[Chapter II, Lemma 7.1]{Satake1980}}]\label{prop:JomegaGrass}
There is a bijection between $\Jomega$ and $\Grass$.
\end{proposition}

\begin{proof} (Sketch)
For any $j\in\Jomega$, $j$ symplectic implies that it is diagonalisable, and $j^2=-I$ implies that its eigenvalues are $\pm i$. The $-i$-eigenspace is simply
\[
\Gamma^j = \frac{1}{2}(I +ij)\mathbb{R}^{2n} \subset \mathbb{C}^{2n}.
\]
It is straightforward to verify that $\Gamma^j\in \Grass$, and so we obtain a map
\[
\Gamma^\cdot:\Jomega \rightarrow \Grass.
\]

Conversely, suppose $\Gamma\in\Grass$. For any $w\in\Gamma-\lbrace 0 \rbrace$, the negative-definite condition implies that $-i\omega(w,\overline{w})<0$. Since $\Gamma$ is Lagrangian, it follows that $\overline{w}\nin \Gamma$. Hence $\Gamma\cap\overline{\Gamma} = \lbrace 0 \rbrace$, and $\mathbb{C}^{2n} = \Gamma \oplus \overline{\Gamma}$. Define the complex-linear map $j_{\Gamma}:\mathbb{C}^{2n}\rightarrow \mathbb{C}^{2n}$ by
\[
j_{\Gamma}(w) = 
\begin{cases} 
-iw & w\in \Gamma \\
iw & w\in\overline{\Gamma}
\end{cases}.
\]
Since it commutes with complex commutation, $j_\Gamma$ restricts to a real-linear map $j_\Gamma:\mathbb{R}^{2n}\rightarrow\mathbb{R}^{2n}$.
It is straightforward to verify that $j_\Gamma\in \Jomega$, and so we obtain a map
\[
j_\cdot: \Grass \rightarrow \Jomega.
\]

It may be checked that $\Gamma^\cdot:\Jomega \rightarrow \Grass$ and $j_\cdot:\Grass\rightarrow \Jomega$ are inverses of each other.
\end{proof}

Following Proposition \ref{prop:JomegaGrass}, we may extend the commutative diagram (\ref{diag:Spfirst}) to
\begin{equation}\label{diag:Spsecond}
\begin{diagram}
&\jO^{-1}(\jO(I)) = \Sptn \\
\ldTo^{\sigma_\mathrm{O}}(1,2)  && \rdTo^{\jSp}(1,2) \rdTo(4,2)^{\Gamma^{\jSp(\cdot)}}\\
N_{\jO(I)} & \rTo^{\cjSp} & \mathcal{O}_{\jSp(I)} = \Jomega\quad & \rTo^{\Gamma^\cdot} &&\Grass
\end{diagram}
\end{equation}

We note that for $S\in\Sptn$,
\[
\Gamma^{\Ad_S j} = \frac{1}{2}(I+i\,\Ad_S j)\mathbb{R}^{2n} = \frac{1}{2}(I+i\,\Ad_S j)S\mathbb{R}^{2n} =S\left(\frac{1}{2}(I+ij)\mathbb{R}^{2n} \right) = S\Gamma^j,
\]
where we have used the invertibility of $S$ to write $\mathbb{R}^{2n} = S\mathbb{R}^{2n}$ in the second equality. Hence $\Gamma^\cdot$ is $\Sptn$-equivariant with respect to the natural left action of $\Sptn$ on $\Grass$.

\subsection{Coordinates on \texorpdfstring{$\Grass$}{the Grassmannian}}
In order to introduce coordinates on $\Grass$, we first introduce the two complex Lagrangian subspaces of $\mathbb{C}^{2n}$
\[
\Gamma_1 = \mathrm{span}_{\mathbb{C}}\lbrace e_1, \ldots, e_n\rbrace \qquad\textrm{and}\qquad \Gamma_2 = \mathrm{span}_{\mathbb{C}}\lbrace e_{n+1},\ldots,e_{2n}\rbrace,
\]
where $\lbrace e_1,\ldots, e_{2n}\rbrace$ is the canonical basis of $\mathbb{R}^{2n}$. Note that $\mathbb{C}^{2n} = \Gamma_1\oplus\Gamma_2$, and denote the direct sum projections by $p_i:\mathbb{C}^{2n}\rightarrow \Gamma_i$. We use the following lemma.

\begin{lemma}\label{lemma:proj}
Let $\Gamma\in \Grass$. Then the projection $p_i:\mathbb{C}^{2n}\rightarrow \Gamma_i$ restricts to a bijection on $\Gamma$, for $i=1,2$.
\end{lemma}

\begin{proof}
Since all three are Lagrangian, $\Gamma, \Gamma_1, \Gamma_2$ all have complex dimension $n$. So we just need to prove that $p_i\vert_{\Gamma}:\Gamma\rightarrow\Gamma_i$ is injective.

Suppose $p_1(w)=0$ for $w\in\Gamma$. Then $w\in\Gamma_2$, and so $\overline{w}\in\overline{\Gamma}_2 = \Gamma_2$. Since $\Gamma_2$ is Lagrangian, we have that $-i\omegacan(w,\overline{w}) = 0$, i.e., $s(w,w) = 0$. But since $s\vert_\Gamma <0$, this implies that $w=0$.

The same argument shows that $p_2\vert_\Gamma:\Gamma\rightarrow \Gamma_2$ is a bijection.
\end{proof}

In light of Lemma \ref{lemma:proj}, for any $\Gamma\in\Grass$ we may define a map $\widehat{W}_\Gamma:\Gamma_1\rightarrow \Gamma_2$ by
\[
\widehat{W}_\Gamma = (p_2\vert_\Gamma)\circ(p_1\vert_\Gamma)^{-1}.
\]
Then for any $w\in\Gamma$, 
\[
w = (p_1+p_2)\vert_\Gamma(w) = (I+\widehat{W}_\Gamma)((p_1\vert_\Gamma)(w)),
\]
demonstrating that 
\begin{equation}\label{eqn:GammaW}
\Gamma = (I+\widehat{W}_\Gamma)\Gamma_1.
\end{equation}

It follows from identity (\ref{eqn:GammaW}) that the map $\Gamma\mapsto \widehat{W}_\Gamma$ is injective. We now use this representation of $\Grass$ as linear maps to introduce coordinates on $\Grass$. The following lemma is straightforward to verify.

\begin{lemma}[{\cite[Chapter II, Lemma 7.2']{Satake1980}}]\label{lemma:SiegelUHP}
Let $\widehat{W}:\Gamma_1\rightarrow \Gamma_2$ be a linear map, and define the matrix $W\in\mathrm{M}_n(\mathbb{C})$ by
\[
\widehat{W} (e_\alpha) = \sum_{\beta=1}^n W_{\beta\alpha}\, e_{n+\beta}\qquad \textrm{for }\alpha=1,\ldots,n.
\]
Then
\[
(I+\widehat{W})\Gamma_1 \in \Grass \iff W^\top = W \textrm{ and }\mathrm{Im}\,W>0.
\]
\end{lemma}

The set 
\[
\Sigma_n = \lbrace W\in\mathrm{M}_{n}(\mathbb{C}) \,|\, W^\top = W,\, \operatorname{Im}W>0 \rbrace
\]
arising in Lemma \ref{lemma:SiegelUHP} is traditionally referred to as the \emph{Siegel upper half plane}. 
It has a well-known symplectic structure. In subsection \ref{sec:sympstruct} we show this structure is the pushforward of the usual Kostant-Kirillov-Souriau symplectic structure on $\mathcal{O}_{\jSp(I)}$ under the bijection $\Gamma^\cdot$ of diagram (\ref{diag:Spsecond}).

\subsection{The action of \texorpdfstring{$\Sptn$}{Sp(2n,R)} in upper half plane coordinates}
There is a natural action of $\Sptn$ on $\Grass$, given by $(S,\Gamma) \mapsto S\Gamma$. We wish to express this action in terms of upper half plane coordinates, i.e., to express $W_{S\Gamma}$ in terms of $S$ and $W_\Gamma$.

Since $\Gamma_1 = \mathrm{span}_{\mathbb{C}}\lbrace e_1,\ldots, e_n\rbrace$, we see from equation (\ref{eqn:GammaW}) that a basis for $\Gamma\in\Grass$ is given by
\[
e_\alpha' := (I+\widehat{W}_\Gamma)e_\alpha = e_\alpha + \sum_{\beta=1}^n(W_\Gamma)_{\beta\alpha}e_{n+\beta},\qquad\alpha=1,\ldots, n.
\]
It follows that $S\Gamma$ has basis
\[
e_\alpha'' = Se_\alpha' = Se_\alpha + \sum_{\beta=1}^n(W_\Gamma)_{\beta\alpha}\, Se_{n+\beta},\qquad\alpha = 1,\ldots, n.
\]
Writing $S=\begin{bmatrix} A & B \\ C & D \end{bmatrix}$, this becomes
\[
e_\alpha'' = \sum_{\gamma=1}^n\left[ (A+BW_\Gamma)_{\gamma\alpha}e_\gamma + (C+DW_\Gamma)_{\gamma\alpha}e_{n+\gamma}\right],\qquad\alpha=1,\ldots,n.
\]

Now recall that the coordinates $W_{S\Gamma}$ of $S\Gamma$ are simply the matrix elements of $\widehat{W}_{S\Gamma} = (p_2\vert_{S\Gamma})\circ (p_1\vert_{S\Gamma})^{-1}:\Gamma_1\rightarrow \Gamma_2$. We have
\begin{align*}
p_1(e_\alpha'') &= \sum_{\gamma=1}^n (A+BW_\Gamma)_{\gamma\alpha}e_\gamma \\
\implies (p_1\vert_{S\Gamma})^{-1}(e_\delta) &= \sum_{\alpha=1}^n[(A+BW_\Gamma)^{-1}]_{\alpha\delta}\,e_\alpha'' \\
\implies  \widehat{W}_{S\Gamma} (e_\delta) &= (p_2\vert_{S\Gamma})\left((p_1\vert_{S\Gamma})^{-1}(e_\delta)\right) = \sum_{\alpha=1}^n[(A+BW_\Gamma)^{-1}]_{\alpha\delta}(p_2(e_\alpha''))\\
&= \sum_{\gamma=1}^n [(C+DW_\Gamma)(A+BW_\Gamma)^{-1}]_{\gamma\delta}\, e_{n+\gamma}.
\end{align*}

It follows that for $S=\begin{bmatrix} A & B \\ C & D\end{bmatrix}$, 
\begin{equation}\label{eqn:WSGamma}
W_{S\Gamma} = (C+DW_\Gamma)(A+BW_\Gamma)^{-1}.
\end{equation}

\subsection{The projection in upper half plane coordinates}
Referring again to diagram (\ref{diag:Spsecond}), we have demonstrated the existence of an $\Sptn$-equivariant map 
\[
\Gamma^{\jSp(\cdot)}:\jO^{-1}(\jO(I))\rightarrow \Grass.
\]
The set $\jO^{-1}(\jO(I))\subset\GLtn$ is simply $\Sptn$, and we have shown that the Siegel upper half plane provides a global coordinate chart on $\Grass$. Using the results of the previous section, we can derive an expression for the projection $\Gamma^{\jSp(\cdot)}$ in terms of the coordinates on its domain and range manifolds.

First note that $\Gamma^{\jSp(I)} = \frac{1}{2}(I+i\jSp(I))\mathbb{R}^{2n} = \frac{1}{2}(I-i\mathbb{J})\mathbb{R}^{2n}$ has a basis
\[
\lbrace e_\alpha + i e_{n+\alpha}\,|\,\alpha=1,\ldots n\rbrace.
\]
It follows easily from equation (\ref{eqn:GammaW}) that 
\[
W_{\Gamma^{\jSp(I)}} = iI\in\mathrm{M}_n(\mathbb{C}).
\]

By equivariance of $\Gamma^{\jSp(\cdot)}$, we have that $\Gamma^{\jSp(S)} = S\Gamma^{\jSp(I)}$ for any $S\in\jSp^{-1}(\jSp(I))$, and so
$W_{\Gamma^{\jSp(S)}} = W_{S\Gamma^{\jSp(I)}}$.
Writing $S = \begin{bmatrix}Q_1 & Q_2 \\ P_1 & P_2\end{bmatrix}$ for the coordinates on $\jSp^{-1}(\jSp(I))$, and using equation (\ref{eqn:WSGamma}), this becomes
\[
W_{\Gamma^{\jSp(S)}} = (P_1+P_2W_{\Gamma^{\jSp(I)}})(Q_1+Q_2W_{\Gamma^{\jSp(I)}})^{-1} = (P_1+iP_2)(Q_1+iQ_2)^{-1}.
\]
In summary, the projection written in (global) coordinates is
\begin{equation}\label{eqn:Gammacoords}
\Gamma^{\jSp(\cdot)}:\begin{bmatrix} Q_1 & Q_2 \\ P_1 & P_2 \end{bmatrix} \mapsto (P_1+iP_2)(Q_1+iQ_2)^{-1}.
\end{equation}

\subsection{Upper half plane coordinates \texorpdfstring{on $\mathcal{O}_{\jSp(I)}$}{the adjoint orbit}}\label{sec:sympstruct}
Using $\Gamma^\cdot:\mathcal{O}_{\jO(I)}\rightarrow \Grass$ to pull coordinates back to $\mathcal{O}_{\jO(I)}$, we may instead interpret (\ref{eqn:Gammacoords}) as the coordinate expression for the momentum map $\jSp:\jO^{-1}(\jO(I))\rightarrow \mathcal{O}_{\jSp(I)}$
\begin{equation}\label{eqn:jSpcoords}
\jSp:\begin{bmatrix} Q_1 & Q_2 \\ P_1 & P_2 \end{bmatrix} \mapsto (P_1+iP_2)(Q_1+iQ_2)^{-1}.
\end{equation}
This is in agreement with \cite[Equation (21)]{Ohsawa2015}, which was obtained by a method involving the Iwasawa decomposition of $\Sptn$.

For $\eta\in\mathcal{O}_{\jSp(I)}$, let $W_\eta\in\Sigma_n$ denote the corresponding coordinate. Then
equation (\ref{eqn:WSGamma}) becomes
\begin{equation}\label{eqn:WAdSeta}
W_{\Ad_S\eta} = (C+DW_\eta)(A+BW_\eta)^{-1}.
\end{equation}

It was shown by Siegel (see for example \cite[Chapter 6, Section 3]{Siegel1973}) that $\Sigma_n$ is a K\"ahler manifold, with Hermitian form $H_{\Sigma_n}=\operatorname{Tr}(\mathcal{B}^{-1}d\overline{W} \otimes \mathcal{B}^{-1}dW)$ and symplectic form $\Omega_{\Sigma_n} = \mathrm{Im}\, H^{\Sigma_n}$, i.e.,
\[
\Omega_{\Sigma_n} = \operatorname{Tr}(\mathcal{B}^{-1} d\mathcal{A}\otimes \mathcal{B}^{-1} d\mathcal{B} - \mathcal{B}^{-1}d\mathcal{B} \otimes \mathcal{B}^{-1} d\mathcal{A}).
\]
(We remind that $\mcA := \operatorname{Re} W,\, \mcB := \operatorname{Im} W$.)
Siegel also proves that the forms $H_{\Sigma_n}$, and so $\Omega_{\Sigma_n}$, are invariant under pullback by the $\Sptn$-action $W\mapsto (C+DW)(A+BW)^{-1}$. 

Correspondingly, the standard Kostant-Kirillov-Souriau symplectic form $\Omega^+_\mathcal{O}$ on the adjoint orbit $\mathcal{O}_{\jSp(I)} = \mathcal{O}_{-\mathbb{J}}$, given by
\[
(\Omega^+_\mathcal{O})_{\eta}(\ad_\xi \eta, \ad_\zeta\eta) = \llangle \eta, [\xi,\zeta]\rrangle,
\]
for $\eta\in\mathcal{O}_{-\mathbb{J}},\ \xi,\zeta\in\sptn$, is invariant under the adjoint $\Sptn$-action. 

Hence, if we can show that $\Omega_{\Sigma_n}$ and $\Omega^+_\mathcal{O}$ are proportional at $-\mathbb{J}$ (corresponding to coordinate $W=iI$), then they are proportional on the entire adjoint orbit $\mathcal{O}_{-\mathbb{J}}$. 

To see this is the case, let $\xi = \begin{bmatrix} \xi_{11} & \xi_{12} \\ \xi_{21} & \xi_{22} \end{bmatrix}\in\sptn$. From (\ref{eqn:WAdSeta}) we deduce
\[
\frac{d}{dt}\bigg\vert_{t=0} W_{\Ad_{\exp t\xi}(-\mathbb{J})} = (\xi_{21}+i\xi_{22}) - i(\xi_{11}+i\xi_{12}),
\]
implying that
\begin{equation}\label{eqn:coordrep}
\ad_\xi(-\mathbb{J}) = (\xi_{12}+\xi_{21})\frac{\partial}{\partial\mathcal{A}} \bigg\vert_{iI} + (\xi_{22}-\xi_{11})\frac{\partial}{\partial\mathcal{B}}\bigg\vert_{iI}.
\end{equation}
Note that $\xi^\top\mathbb{J} + \mathbb{J}\xi = 0 \implies \xi_{12}^\top = \xi_{12},\, \xi_{21}^\top = \xi_{21},\,\xi_{22}^\top = -\xi_{11}$, and so both coefficient matrices in \eqref{eqn:coordrep} are symmetric.

We have that
\begin{align}\label{eqn:OmegaO}
(\Omega^+_\mathcal{O})_{-\mathbb{J}} & (\ad_\xi(-\mathbb{J}), \ad_\zeta(-\mathbb{J})) = \llangle -\mathbb{J}, [\xi,\zeta]\rrangle \nonumber \\
&\qquad= \frac{1}{2}\operatorname{Tr}\left( -(\xi_{21}\zeta_{11} + \xi_{22}\zeta_{21}) + (\xi_{11}\zeta_{12} + \xi_{12}\zeta_{22}) - (\xi\leftrightarrow\zeta) \right).
\end{align}
Meanwhile, at $W=iI$ the symplectic form on $\Sigma_n$ simplifies to
\[
(\Omega_{\Sigma_n})_{iI} = \operatorname{Tr}(d\mathcal{A}\otimes d\mathcal{B} - d\mathcal{B}\otimes d\mathcal{A})_{iI},
\]
and so (\ref{eqn:coordrep}) gives
\[
(\Omega_{\Sigma_n})_{iI}(\ad_\xi(-\mathbb{J}), \ad_\zeta(-\mathbb{J})) = \operatorname{Tr}((\xi_{12}+\xi_{21})(\zeta_{22}-\zeta_{11}) - (\xi_{22}-\xi_{11})(\zeta_{12}+\zeta_{21})).
\]
Comparison with (\ref{eqn:OmegaO}) shows that $(\Omega^+_\mathcal{O})_{-\mathbb{J}} = \frac{1}{2}(\Omega_{\Sigma_n})_{iI}$, and therefore
\[
\Omega^+_\mathcal{O} = \frac{1}{2}\Omega_{\Sigma_n}.
\]

\section{Dynamics}\label{sec:dynamics}

In this section we finally come to the main point of the previous constructions, namely to describe the dynamics of semiclassical Gaussian wave packets in terms of Hamiltonian dynamics on the frame bundle of $T^*\mathbb{R}^n$, and its symplectic reduction.

\subsection{The Gaussian wave packet ansatz}
Consider the time-dependent Schr\"odinger equation on $\mathbb{R}^n$,
\begin{equation}\label{eqn:TDSE}
i\hbar\frac{\partial\psi}{\partial t}(x,t) = -\frac{\hbar^2}{2m}\Delta\psi(x,t) + V(x)\,\psi(x,t),
\end{equation}
where the potential $V(x)$ is at most \emph{quadratic} in $x$.
Following \cite{Lubich2008}, we recall here the Gaussian wave packet ansatz, in both the Heller \cite{Heller1976} and Hagedorn \cite{Hagedorn1980} parametrisations. 

In Heller's parametrisation, the wavefunction is written
\begin{equation}\label{eqn:Helleransatz}
\begin{split}
\psi(x,t) = \left(\frac{\det \mcB(t)}{(\pi\hbar)^n} \right)^{\frac{1}{4}} \exp \bigg\lbrace \frac{i}{\hbar} \bigg[&\frac{1}{2}(x-q(t))^\top (\mcA(t)+i\mcB(t))(x-q(t)) \\
&+ p(t)^\top(x-q(t)) + \phi(t) \bigg] \bigg\rbrace,
\end{split}
\end{equation}
with $(q,p)\in T^*\mathbb{R}^n\simeq \mathbb{R}^{2n}$, $\phi$ a real phase, and $\mcA+i\mcB\in \Sigma_n$ with $\mcA,\mcB$ real.
Substitution of the ansatz into equation (\ref{eqn:TDSE}) shows that it produces a solution provided the following equations are satisfied
\begin{align}\label{eqn:Hellereqns}
\dot{q} &= \frac{p}{m}, \qquad \dot{p} = -D_q V, \nonumber \\
\dot{\mcA} &= -\frac{1}{m}(\mathcal{A}^2 - \mathcal{B}^2) - D^2_q V, \qquad \dot{\mcB}= -\frac{1}{m}(\mcA \mcB + \mcB \mcA),\\
\dot{\phi} &= \frac{p^2}{2m} - V(q) - \frac{\hbar}{2m}\operatorname{Tr}\mcB \nonumber.
\end{align}
Here, $D^2V$ denotes the Hessian of $V$.

In Hagedorn's parametrisation, the wavefunction is
\begin{equation}\label{eqn:Hagedornansatz}
\begin{split}
\psi(x,t) =  \frac{1}{(\pi\hbar)^{\frac{n}{4}}} \frac{1}{(\det Q(t))^{\frac{1}{2}}} \exp \bigg\lbrace \frac{i}{\hbar} \bigg[&\frac{1}{2}(x-q(t))^\top P(t)Q(t)^{-1}(x-q(t)) \\
&+ p(t)^\top(x-q(t)) + S(t) \bigg] \bigg\rbrace,
\end{split}
\end{equation}
for some appropriate choice of the square root in the denominator, with $(q,p)\in\mathbb{R}^{2n}$, $S$ a real phase, and $Q, P\in\mathrm{M}_n(\mathbb{C})$ satisfying 
\[
Q^\top P - P^\top Q = 0 \qquad\textrm{and}\qquad Q^\dagger P - P^\dagger Q = 2iI_n.
\]
Substitution into equation (\ref{eqn:TDSE}) yields
\begin{equation}\label{eqn:Hagedorneqns}
\dot{q} = \frac{p}{m}, \qquad\dot{p} = -D_q V, \qquad \dot{Q} = \frac{P}{m}, \qquad \dot{P} = -D^2_q V\,Q, \qquad \dot{S} = \frac{p^2}{2m} - V(q). 
\end{equation}

The two parametrisations (\ref{eqn:Helleransatz}) and (\ref{eqn:Hagedornansatz}) are related by
\begin{equation}\label{eqn:coordrelation}
W = PQ^{-1}, \qquad S = \phi - \frac{\hbar}{2}\mathrm{arg}(\det Q).
\end{equation}

\subsection{Description as Hamiltonian dynamics on the frame bundle}
Given a symplectic manifold $(M,\omega)$, recall that the a \emph{frame} at $z\in M$ is an ordered basis of $T_z M$, and the \emph{frame bundle} $\mathcal{F}(M)$ consists of all such ordered bases as $z$ ranges over $M$. $\mathcal{F}(M)$ is naturally a right principal $\GLtn$-bundle, with group action
\begin{equation}\label{eqn:GLnaction}
(v_1,\ldots, v_{2n})\cdot E = \left(\sum_{a=1}^{2n} v_a E_{a1}, \sum_{a=1}^{2n} v_a E_{a2}, \ldots, \sum_{a=1}^{2n} v_a E_{a,2n} \right). 
\end{equation}
We say that a frame $(v_a)$ at $z\in M$ is symplectic if
\[
\omega_z(v_a,v_b) = \mathbb{J}_{ab}\qquad a,b=1,\ldots,2n,
\]
and a local section $\bm{s}:U\subset M\rightarrow \mathcal{F}(M)$ of the frame bundle is symplectic if $\bm{s}(z)$ is a symplectic frame for each $z\in U$.

Consider now the case when $M=T^*\mathbb{R}^n$ and $\omega = \sum_{\alpha=1}^n dq^\alpha\wedge dp_\alpha$ is the usual cotangent bundle symplectic form. We will write the natural coordinates on $T^*\mathbb{R}^n$ as either $(z_1,\ldots, z_{2n})$ or $(q^1,\ldots, q^n, p_1,\ldots, p_n)$, depending on situation. The bundle $\mathcal{F}(T^*\mathbb{R}^n)$ has a \emph{global} symplectic frame $\bm{s}:T^*\mathbb{R}^n\rightarrow \mathcal{F}(T^*\mathbb{R}^n)$, given by
\[
\bm{s}(z) := \left(\frac{\partial}{\partial q^1}\bigg\vert_z, \ldots, \frac{\partial}{\partial q^n}\bigg\vert_z,\frac{\partial}{\partial p_1}\bigg\vert_z, \ldots,\frac{\partial}{\partial p_n}\bigg\vert_z\right) = \left(\frac{\partial}{\partial z_1}\bigg\vert_z, \ldots, \frac{\partial}{\partial z_{2n}}\bigg\vert_z\right),
\]
which induces a global trivialisation $\Lambda:T^*\mathbb{R}^n \times \GLtn \rightarrow \mathcal{F}(T^*\mathbb{R}^n)$ of the frame bundle, given by
\[
\Lambda(z,E) = \bm{s}(z)\cdot E = \left(\sum_{a=1}^{2n}E_{a1}\frac{\partial}{\partial z_a}\bigg\vert_z, \ldots, \sum_{a=1}^{2n}E_{a,2n}\frac{\partial}{\partial z_a}\bigg\vert_z \right).
\]

We define the \emph{$\hbar$-diagonal lifted} symplectic form $\Omega^{\hbar}$ on $T^*\mathbb{R}^n\times\GLtn$ by
\begin{equation}\label{eqn:Omegah}
\Omega^{\hbar}_{(z,E)}((v_z, V_E), (w_z,W_E)) := \omega_z(v_z,w_z) + \frac{\hbar}{2}\Omega_E(V_E,W_E).
\end{equation}
By analogy with \cite{Ohsawa2015}, for any Hamiltonian $H:T^*\mathbb{R}^n\rightarrow \mathbb{R}$ we define the lifted Hamiltonian on the frame bundle $H^{\hbar}:T^*\mathbb{R}\times\GLtn \rightarrow \mathbb{R}$ to be
\begin{equation}\label{eqn:liftHam1}
H^{\hbar}(z,E) := H(z) + \frac{\hbar}{4}\operatorname{Tr}(E^\top D^2_z H E),
\end{equation}
where $D^2_zH$ denotes the Hessian of $H$ evaluated at $z$. 

We can give a more geometric interpretation to (\ref{eqn:liftHam1}) as follows: for any $E\in \GLtn$, let $g(E)$ denote the metric on $T^*\mathbb{R}^n$ with respect to which the global frame $\Lambda(z,E) = \bm{s}(z)\cdot E$ is orthonormal, i.e., the vector fields
\[
e_a(z) = \sum_{b=1}^{2n}E_{ba}\frac{\partial}{\partial z_b},\qquad a=1,\ldots,2n,
\]
form an orthonormal basis at each point. Explicitly
\begin{equation}\label{eqn:unreducedmetric}
g(E)\left(\frac{\partial}{\partial z_a}, \frac{\partial}{\partial z_b}\right) = [(EE^\top)^{-1}]_{ab}.
\end{equation}
Then the Laplacian $\Delta_{g(E)}:C^\infty(T^*\mathbb{R}^n)\rightarrow C^\infty(T^*\mathbb{R}^n)$ associated with the metric $g(E)$ is given for $H\in C^\infty(T^*\mathbb{R}^n)$ by
\[
\Delta_{g(E)}H = \sum_{a=1}^{2n} e_a (e_a H) = \sum_{a,b,c=1}^{2n}E_{ba}E_{ca}\frac{\partial^2 H}{\partial z_b\partial z_c} = \operatorname{Tr}(E^\top D^2H E),
\]
and so $(\ref{eqn:liftHam1})$ can be alternatively written as
\begin{equation}\label{eqn:liftHam2}
H^{\hbar}(z,E) = H(z) + \frac{\hbar}{4}(\Delta_{g(E)}H)(z).
\end{equation}

To find the Hamiltonian vector field corresponding to $H^{\hbar}$, note that
\[
d_{(z,E)}H^{\hbar} = (d_1)_z \left( H + \frac{\hbar}{4} \Delta_{g(E)}H\right) + (d_2)_E \left(\frac{\hbar}{4}\Delta_{g(\cdot)}H(z)\right),
\]
where $d_1, d_2$ denote the exterior derivatives in $T^*\mathbb{R}^n$ and $\GLtn$ respectively.
The first term can be written
\[
i_{X_{H+ \frac{\hbar}{4}\Delta_{g(E)}H}}\omega.
\]
The second term satisfies
\begin{align*}
(d_2)_E\left(\frac{\hbar}{4}\Delta_{g(\cdot)}H(z)\right) (V_E) &= \frac{\hbar}{4} \operatorname{Tr}\left(V^\top D^2_zH\,E + E^\top D^2_zH\, V \right) \\
&= \frac{\hbar}{2}\operatorname{Tr}\left( E^\top D^2_zH\,V \right)\\
&= \frac{\hbar}{2}\Omega_E((\mathbb{J}D^2_z H)\cdot E, V_E) \\
&= \frac{\hbar}{2}\left( i_{(\mathbb{J}D^2_zH)\cdot E}\Omega_E \right)(V_E).
\end{align*}
Here we have used the identity $\zeta\cdot E = (\zeta E)_E$. Overall
\[
d_{(z,E)} H^{\hbar} = i_{(X_{H+\frac{\hbar}{4}\Delta_{g(E)}H}(z), (\mathbb{J}D^2_ZH)\cdot E)}\Omega^{\hbar}_{(z,E)},
\]
i.e.,
\begin{equation}\label{eqn:XHhbar}
X_{H^{\hbar}}(z,E) = \left( X_{H+\frac{\hbar}{4}\Delta_{g(E)}H}(z), (\mathbb{J}D^2_zH)\cdot E \right).
\end{equation}

In the situation where $H:T^*\mathbb{R}^n\rightarrow \mathbb{R}$ is at most quadratic in the coordinates $(q^1,\ldots, q^n,p_1,\ldots, p_n)$, the Hessian $D^2_zH$, and hence $\Delta_{g(E)}H$, will be constant, and the Hamiltonian vector field $X_{H^{\hbar}}$ simplifies to
\[
X_{H^{\hbar}}(z,E) = \left( X_H(z), (\mathbb{J}D^2_zH)\cdot E \right) = \left( \mathbb{J}D_zH, (\mathbb{J}D^2_zH)\cdot E \right) 
\]
In this case, motion on the frame bundle consists of classical motion on the base space $T^*\mathbb{R}^n$, and the corresponding induced linearised motion on the frame. This picture of semiclassical wave mechanics has much in common with the \emph{nearby orbit approximation} discussed by Littlejohn \cite[Section 7]{Littlejohn1986}\textemdash semiclassical motion consists of the classical motion, plus motion of a frame moving along the classical trajectory, and describing to first order the classical flow relative to it.

In particular, taking the standard Hamiltonian
\[
H(q,p) = \frac{p^2}{2m} + V(q)
\]
and writing $E = \begin{bmatrix} Q_1 & Q_2 \\ P_1 & P_2 \end{bmatrix}$
yields the first four equations of the Hagedorn equations (\ref{eqn:Hagedorneqns}).

In the general case (\ref{eqn:XHhbar}) where $H$ is not quadratic, there is an additional correction term in the motion on the base space, leading to a deviation from strictly classical motion. This deviation has recently been proposed as a means of introducing quantum tunnelling into semiclassical quantum mechanics \cite{OhsawaLeok2013}.

\subsection{\texorpdfstring{$\Otn$}{O(2n)}-invariance and conservation of symplectic frame}
The right $\GLtn$-action (\ref{eqn:GLnaction}) on the bundle $\mathcal{F}(T^*\mathbb{R}^n) \simeq T^*\mathbb{R}^n\times\GLtn$ restricts to a right $\Otn$-action, and the lifted Hamiltonian $H^{\hbar}$ is easily seen to be invariant under this restricted action. By Noether's theorem, the corresponding momentum map $\jO\circ \pi_2:T^*\mathbb{R}^n\times \GLtn\rightarrow \otn$ is conserved under the dynamical evolution corresponding to $H^{\hbar}$, i.e.,
\[
\jO(E) = -E^\top \mathbb{J} E
\]
is conserved. In particular, the previously considered level set
\[
\jO^{-1}(\jO(I)) = \lbrace E\in \GLtn\,|\,-E^\top\mathbb{J}E = -\mathbb{J} \rbrace
\]
simply corresponds to the collection of \emph{symplectic} frames at a point, and conservation of $\jO$ says that such a frame remains symplectic under dynamical evolution.

\subsection{Description in the reduced space}
Since we have a Hamiltonian $H^{\hbar}$ that is invariant under the right $\Otn$-action, the dynamics on $\mathcal{F}(T^*\mathbb{R}^n)$ will drop to the symplectic quotient. Taking the symplectic quotient at momentum $\jO(I) = -\mathbb{J}$, and using the notation of diagram (\ref{diag:Spsecond}), gives
\[
\frac{(\jO\circ\pi_2)^{-1}(\jO(I))}{\Otn_{\jO(I)}} = T^*\mathbb{R}^n \times \frac{\jO^{-1}(\jO(I))}{\Otn_{\jO(I)}} = T^*\mathbb{R}^n\times N_{\jO(I)},
\]
and we have the corresponding extension of diagram (\ref{diag:Spfirst})
\begin{diagram}
&&T^*\mathbb{R}^n\times\jO^{-1}(\jO(I)) \\
&\ldTo^{\operatorname{id}\times\sigma_\mathrm{O}} && \rdTo^{\operatorname{id}\times\jSp}\\
T^*\mathbb{R}^n\times N_{\jO(I)} && \rTo^{\operatorname{id}\times\cjSp} && T^*\mathbb{R}^n\times \mathcal{O}_{\jSp(I)}.
\end{diagram}
The map $\operatorname{id}\times\cjSp$ is a symplectomorphism, and so we can equivalently describe the reduced dynamics on $T^*\mathbb{R}^n\times\mathcal{O}_{\jSp(I)}$. Rewriting equation (\ref{eqn:liftHam1}) as
\[
H^{\hbar}(z,E) = H(z) + \frac{\hbar}{4}\operatorname{Tr}(-EE^\top \mathbb{J}^2 D^2_z H) = H(z) + \frac{\hbar}{4}\operatorname{Tr}(\jSp(E)\,\mathbb{J}D^2_zH)
\]
we see that 
\[
H^{\hbar} = \check{H}^{\hbar}\circ(\operatorname{id}\times\jSp), 
\]
where $\check{H}^{\hbar}: T^*\mathbb{R}^n\times\mathcal{O}_{\jSp(I)}\rightarrow \mathbb{R}$ is the reduced Hamiltonian
\begin{equation}\label{eqn:redliftHam1}
\check{H}^{\hbar}(z,\zeta) := H(z) + \frac{\hbar}{4}\operatorname{Tr}(\zeta\, \mathbb{J}D^2_zH).
\end{equation}
In fact, from the expression (\ref{eqn:unreducedmetric}) for the metric $g(E)$, we see that 
\[
g(E)\left(\frac{\partial}{\partial z_a},\frac{\partial}{\partial z_b}\right) = [(\jSp(E)\mathbb{J})^{-1}]_{ab},
\]
and so the metric $g(E)$ drops to a metric $\check{g}(\zeta)$ on $T^*\mathbb{R}^n$, depending on $\zeta\in\sptn$,
\[
\check{g}(\zeta)\left(\frac{\partial}{\partial z_a},\frac{\partial}{\partial z_b} \right) = [(\zeta\mathbb{J})^{-1}]_{ab}.
\]
Then the reduced Hamiltonian \eqref{eqn:redliftHam1} can be written, analogously to \eqref{eqn:liftHam2}, as
\begin{equation}\label{eqn:redliftHam2}
\check{H}^{\hbar}(z,\zeta) := H(z) + \frac{\hbar}{4} (\Delta_{\check{g}(\zeta)}H)(z).
\end{equation}

Since $H^{\hbar} = \check{H}^{\hbar}\circ(\operatorname{id}\times \jSp)$, and $\operatorname{id}\times\jSp$ is a Poisson map, the Hamiltonian vector fields 
\[
X_{H^{\hbar}}(z,E) = \left(X_{H+ \frac{\hbar}{4}\Delta_{g(E)} H}(z), (\mathbb{J} D^2_zH)\cdot E \right)
\]
and 
\[
X_{\check{H}^{\hbar}}(z,\zeta) = \left(X_{H+\frac{\hbar}{4}\Delta_{\hat{g}(\zeta)}H}(z), \ad_{\mathbb{J} D^2_z H} \zeta \right)
\] 
are $(\operatorname{id} \times \jSp)$-related,
\[
T(\operatorname{id}\times \jSp)\circ X_{H^{\hbar}} = X_{\check{H}^{\hbar}}\circ(\operatorname{id}\times\jSp).
\]
Again specialising to quadratic Hamiltonians, these $(\operatorname{id}\times\jSp)$-related vector fields are
\[
X_{H^{\hbar}}(z,E) = (X_H(z), (\mathbb{J}D^2_zH)\cdot E)
\]
and
\[
X_{\check{H}^{\hbar}}(z,\zeta) = (X_H(z), \ad_{\mathbb{J}D^2_zH}\zeta).
\]


The unreduced dynamics, generated by $X_{H^{\hbar}}$ and reproducing the the Hagedorn equations (\ref{eqn:Hagedorneqns}), will drop to the reduced dynamics generated by $X_{\check{H}^{\hbar}}$, and these will reproduce the Heller equations (\ref{eqn:Hellereqns}), by virtue of the relationship \eqref{eqn:jSpcoords} expressing $\jSp$ in coordinates on $\jO^{-1}(\jO(I))$ and $\mathcal{O}_{\jSp(I)}$.

\section{Generalisation to other symplectic manifolds}\label{sec:deLeon}
Now that we have outlined the construction of a Hamiltonian system on the frame bundle of $T^*\mathbb{R}^n$, an obvious question is whether this construction can be extended to more general symplectic manifolds. In this section we suggest a generalisation by employing the results of Cordero and de Le\'on \cite{CorderodeLeon1983}. Essentially they show that the frame bundle can be provided with a symplectic structure if it has a trivialisation compatible with the symplectic structure, encoded in the form of a symplectic connection. Their construction may be seen as the symplectic analogue of the Sasaki-Mok metric on the frame bundle of a Riemannian manifold \cite{Mok1978}.

\subsection{The almost-symplectic form on the frame bundle}
Let $(M,\omega)$ be an arbitrary $2n$-dimensional symplectic manifold, and let $\pi:\mathcal{F}(M)\rightarrow M$ denote its frame bundle, which is a principal right $\GLtn$-bundle in the usual way. For a frame $\bm{e} = (e_1,\ldots, e_{2n}) \in\mathcal{F}(M)$, and $\xi\in\gltn$, define
\[
\bm{e}\cdot\xi := \frac{d}{ds}\bigg\vert_{s=0} \bm{e}\cdot \exp(s\xi) \in V_{\bm{e}}\mathcal{F}(M) \subset T_{\bm{e}}\mathcal{F}(M),
\]
where $V_{\bm{e}}$ denotes the vertical space at $\bm{e}\in\mathcal{F}(M)$. Given a connection 1-form $A\in\Omega^2(\mathcal{F}(M),\gltn)$, Cordero and de Le\'on \cite{CorderodeLeon1983} define two natural lifts of the symplectic 2-form $\omega$ to the frame bundle $\mathcal{F}(M)$:
\begin{itemize}
\item
the vertical lift $\omega^{\mathrm{ver}} = \pi^*\omega$;
\item
the horizontal lift $\omega^{\mathrm{hor}}$, given by
\[
\omega^{\mathrm{hor}}_{\bm{e}} (X_{\bm{e}}, Y_{\bm{e}}) := \left(\bigoplus_{a=1}^{2n} \omega \right)_{\bm{e}} (\bm{e}\cdot A_{\bm{e}}(X_{\bm{e}}), \bm{e}\cdot A_{\bm{e}}(X_{\bm{e}})).
\]
\end{itemize}
In the definition of the horizontal lift we have used the canonical isomorphism $\oplus_{a=1}^{2n}T_z M \simeq V_{\bm{e}}\mathcal{F}(M)$ (for $z=\pi(\bm{e}))$ given by
\[
(v_1, \ldots, v_{2n}) \in \bigoplus_{a=1}^{2n} T_z M \longleftrightarrow \frac{d}{ds}\Bigg\vert_{s=0}(e_1 + sv_1,\ldots e_{2n}+sv_{2n})\in V_{\bm{e}}\mathcal{F}(M),
\]
and here $\bm{e}\cdot A_{\bm{e}}(X_{\bm{e}})$ is the infinitesimal generator at $\bm{e} \in \mathcal{F}(M)$ corresponding to $A_{\bm{e}}(X_{\bm{e}})\in\gltn$.
We then define the \emph{$\hbar$-diagonal lift} $\omega^{\hbar}$ of the symplectic form $\omega$ by
\[
\omega^{\hbar} := \omega^{\mathrm{ver}} + \frac{\hbar}{2}\omega^{\mathrm{hor}}.
\]
This is essentially the diagonal lift defined in \cite{CorderodeLeon1983}, up to a factor of $\frac{\hbar}{2}$. 

Denote the horizontal lift of a vector field $X\in\mathfrak{X}(M)$ to $\mathcal{F}(M)$ by $X^{\mathrm{hor}}$, and the infinitesimal generator on $\mathcal{F}(M)$ corresponding to $\xi\in\gltn$ by $\xi_{\mathcal{F}(M)}$. Then it is straightforward to verify
\begin{itemize}
\item
$\omega^{\hbar}(X^{\mathrm{hor}}, Y^{\mathrm{hor}}) = \omega(X,Y)\circ\pi$,
\item
$\omega^{\hbar}(Z^{\mathrm{hor}},\xi_{\mathcal{F}(M)}) = 0$, and
\item
$\omega^{\hbar}(\xi_{\mathcal{F}(M))}, \zeta_{\mathcal{F}(M)})(\bm{e}) = \frac{\hbar}{2}\sum_{a=1}^{2n}\omega_z((\bm{e}\cdot \xi)_a, (\bm{e}\cdot\zeta)_a) = \frac{\hbar}{2}\sum_{a,b,c=1}^{2n}\omega_z(e_b,e_c)\,\xi_{ba}\zeta_{ca}$.
\end{itemize}
Using the non-degeneracy of $\omega$ to infer that $\omega_z(e_b,e_c)$ is an invertible matrix in the last identity, it follows that the 2-form $\omega^{\hbar}$ is non-degenerate, and so defines an almost-symplectic form on $\mathcal{F}(M)$.

\subsection{Conditions for the lifted form to be closed}
In order for the almost-symplectic form $\omega^{\hbar}$ to be a symplectic form, it must satisfy the additional condition $d\omega^{\hbar} = 0$. The conditions for this to be the case are computed in \cite{CorderodeLeon1983}.

\begin{proposition}[{\cite[Proposition 6.1]{CorderodeLeon1983}}]
\leavevmode
\begin{enumerate}[(i)]
\item
$d\omega^{\hbar}(X^{\mathrm{hor}}, Y^{\mathrm{hor}}, Z^{\mathrm{hor}}) = (d\omega(X,Y,Z))\circ\pi$;
\item
$d\omega^{\hbar}(X^{\mathrm{hor}}, Y^{\mathrm{hor}}, \xi_{\mathcal{F}(M)}) = \frac{\hbar}{2} \omega^{\mathrm{ver}}(F^A(X^{\mathrm{hor}}, Y^{\mathrm{hor}})_{\mathcal{F}(M)}, \xi_{\mathcal{F}(M)})$;
\item
$d\omega^{\hbar}(X^{\mathrm{hor}},\xi_{\mathcal{F}(M)}, \zeta_{\mathcal{F}(M)}) = \frac{\hbar}{2}(\nabla^A_X\omega)^{\mathrm{ver}}(\xi_{\mathcal{F}(M)},\zeta_{\mathcal{F}(M)})$;
\item
$d\omega^{\hbar}(\xi_{\mathcal{F}(M)}, \zeta_{\mathcal{F}(M)},\chi_{\mathcal{F}(M)}) = 0$.
\end{enumerate}
Here, $F^A\in\Omega^2(\mathcal{F}(M),\gltn)$ denotes the curvature 2-form of $A$, and $\nabla^A$ denotes the covariant derivative associated with $A$.
\end{proposition}

\begin{corollary}\label{cor:omegahsymp}
\[
d\omega^{\hbar} = 0 \iff d\omega = 0, \quad F^A = 0, \quad \nabla^A\omega = 0,
\]
i.e., $\omega^{\hbar}$ is closed if and only if $\omega$ is closed and $A$ is a flat, symplectic (i.e. $\omega$-preserving) connection.
\end{corollary}

\subsection{Consequences of the condition}
Assuming the conditions in Corollary \ref{cor:omegahsymp} hold, the frame bundle $\mathcal{F}(M)$ has a particularly simple description. Firstly, $F^A=0$ implies that the horizontal distribution on $\mathcal{F}(M)$ is involutive, and hence integrable. If we further assume the absense of monodromy\footnote{This is the case for example if $M$ is simply connected\textemdash if not we can always go to universal cover.}, this implies the existence of a global horizontal section $\bm{s}:M\rightarrow \mathcal{F}(M)$ through any point of $\mathcal{F}(M)$. Since the connection $A$ is symplectic, the frames $\bm{s}(z) = (s_1(z),\ldots, s_{2n}(z))$ are such that $\omega_z(s_a(z), s_b(z))$ is independent of $z$. Let us arrange for these frames to be symplectic, i.e.,
\[
\omega_z(s_a(z),s_b(z)) = \mathbb{J}_{ab}.
\]
Then $\mathcal{F}(M)$ is trivial, and we have a bundle morphism $\Lambda:M\times\GLtn\rightarrow \mathcal{F}(M)$ given by
\[
\Lambda(z,E) = \bm{s}(z)\cdot E.
\]
We use $\Lambda$ to pull back the symplectic form $\omega^{\hbar}$. Firstly,
\[
\Omega^{\mathrm{ver}} := \Lambda^*\omega^{\mathrm{ver}} = (\pi\circ\Lambda)^*\omega = \pi_1^*\omega,
\]
where $\pi_1:M\times\GLtn\rightarrow M$ is projection onto the first factor. Also, defining $\Omega^{\mathrm{hor}} := \Lambda^*\omega^{\mathrm{hor}}$, and using
\[
T_{(z,E)}\Lambda (v_z, V_E) = T_z \bm{s}(v_z)\cdot E + \bm{s}(z) \cdot V
\]
we get
\begin{align*}
\Omega^{\mathrm{hor}}_{(z,E)} ((v_z, V_E), (w_z,W_E)) &= \omega^{\mathrm{hor}}_{\bm{s}(z)\cdot E} ( \bm{s}(z)\cdot V, \bm{s}(z)\cdot W )\\  
&\qquad\qquad \textrm{since $T_z\bm{s}(v_z)\cdot E$ is horizontal}\\
&= \sum_{a,b,c=1}^{2n}\omega_z((s_b(z),s_c(z))V_{ba}W_{ca} \\
&= \operatorname{Tr}(V^\top \mathbb{J} W).
\end{align*}
Thus we have shown
\[
\Lambda^*\omega^{\mathrm{hor}} =: \Omega^{\mathrm{hor}} = \pi_2^*\Omega,
\]
where $\pi_2: M\times\GLtn \rightarrow \GLtn$ is projection onto the second factor, and $\Omega$ is the symplectic 2-form on $\GLtn$ defined in Section \ref{sec:sympMn}.

Overall, we obtain the symplectic form $\Omega^{\hbar} := \Lambda^*\omega^{\hbar}$ on $M\times \GLtn$, given explicitly by
\[
\Omega^{\hbar} = \Omega^{\mathrm{ver}} + \frac{\hbar}{2}\Omega^{\mathrm{ver}} = \pi_1^*\omega + \frac{\hbar}{2}\pi_2^*\Omega = \omega \oplus \frac{\hbar}{2}\Omega,
\]
i.e.,
\begin{equation}\label{eqn:Omegahgen}
\begin{split}
\Omega^{\hbar}_{(z,E)}((v_z,V_E), (w_z,W_E)) &= \omega_z(v_z,w_z) + \frac{\hbar}{2}\Omega_E(V_E, W_E) \\
&= \omega_z(v_z,w_z) + \frac{\hbar}{2} \operatorname{Tr}(V^\top \mathbb{J} W).
\end{split}
\end{equation}
We see that equation (\ref{eqn:Omegahgen}) agrees with equation (\ref{eqn:Omegah}), which corresponds to applying the above construction with $A$ the obvious global flat connection on $T^*\mathbb{R}^n$, and global horizontal section
\[
\bm{s}(z) = \left(\frac{\partial}{\partial q^1}\bigg\vert_z, \ldots, \frac{\partial}{\partial q^n}\bigg\vert_z,\frac{\partial}{\partial p_1}\bigg\vert_z, \ldots,\frac{\partial}{\partial p_n}\bigg\vert_z\right).
\]

\subsection{The frame bundle Hamiltonian}
Again assuming the conditions in Corollary \ref{cor:omegahsymp} hold, and given a Hamiltonian $H:M\rightarrow \mathbb{R}$, definition (\ref{eqn:liftHam2}) of the lifted Hamiltonian carries over directly. To recap, a choice of connection $A\in\Omega^1(\mathcal{F}(M),\gltn)$ and global horizontal section $\bm{s}:M \rightarrow \mathcal{F}(M)$ induces a global bundle morphism $\Lambda(z,E) = \bm{s}(z)\cdot E$. For every $E\in\GLtn$, we define $g(E)$ to be the metric on $M$ with respect to which the frame $\Lambda(z,E)$ is orthonormal. The lifted Hamiltonian $H^{\hbar}:M\times\GLtn\rightarrow \mathbb{R}$ is then defined by
\[
H^{\hbar}(z,E) := H(z) + \frac{\hbar}{4}\left(\Delta_{g(E)}H\right)(z).
\]
As before, this may also be written
\[
H^{\hbar}(z,E) := H(z) + \frac{\hbar}{4} \operatorname{Tr} (E^\top  D^2_z H E),
\]
where now $D^2H:M\rightarrow\gltn$ is given by
\[
(D^2_zH)_{ab} := s_a(z) \left( s_b H \right).
\]
Note however that by constrast with equation (\ref{eqn:liftHam1}), here $D^2H$ is not necessarily symmetric, since $s_b$ is also a function of $z$. Using a similar analysis to that in Section \ref{sec:dynamics} gives the Hamiltonian vector field generated by $H^{\hbar}$ as
\begin{equation}\label{eqn:XHhbarnonsymm}
X_{H^{\hbar}} (z,E) = \left( X_{H+\frac{\hbar}{4}\Delta_{g(E)}}(z), \frac{1}{2}\mathbb{J}(D^2_zH + (D^2_zH)^\top )\cdot E \right).
\end{equation}

Symmetry of $D^2H$ may be recovered by imposing a further condition on the connection $A$, viz. that it be \emph{torsionless}. For then, applying the torsionless condition to the individual (covariantly constant) vectors of the global frame $\bm{s}$ yields
\[
0 = \nabla^A_{s_a} s_b - \nabla^A_{s_b} s_a - [s_a, s_b] \implies [s_a, s_b] = 0,
\]
and it follows that
\[
(D^2_zH)_{ab} = s_a(z) \left( s_b H\right) = s_b(z) \left( s_a H\right) = (D^2_zH)_{ba}.
\]
Hence equation (\ref{eqn:XHhbarnonsymm}) reduces to equation (\ref{eqn:XHhbar}).

From here on, the analysis is identical to Section \ref{sec:dynamics}. In particular, $\Lambda$ intertwines the right $\Otn$-actions on the frame bundle $\mathcal{F}(M)$ and $M\times\GLtn$
\[
\Lambda(z,E)\cdot O = (\bm{s}(z)\cdot E)\cdot O = \bm{s}(z)\cdot(EO) = \Lambda(z,EO),
\]
and the latter action has a momentum map $\jO:M\times\GLtn\rightarrow \otn$ as before. Also the ($\bm{s}$-dependent) left $\Sptn$-action
\[
S\cdot (\bm{s}(z)\cdot E) = \bm{s}\cdot (SE),
\]
is intertwined with the left $\Sptn$-action on $M\times \GLtn$, and gives a momentum map $\jSp:M\times\GLtn\rightarrow \sptn$ as before.

\section*{Acknowledgements}
The author is grateful to Tomoki Ohsawa, Cesare Tronci, Manuel de Le\'on, Cornelia Vizman, Darryl Holm, and Martin Wolf for useful discussions. This work was supported by Leverhulme Trust Research Project Grant 2014-112, and the Institute of Mathematics and its Applications Small Grant Scheme.


\end{document}